\documentclass[lettersize,journal]{IEEEtran}
\usepackage{amsfonts}
\usepackage{array}
\usepackage[caption=false,font=normalsize,labelfont=sf,textfont=sf]{subfig}
\usepackage{textcomp}
\usepackage{stfloats}
\usepackage{url}
\usepackage{verbatim}
\usepackage{graphicx}
\usepackage{cite}
\usepackage{amssymb}
\usepackage{multirow}
\usepackage{pifont} 
\usepackage{utfsym} 
\usepackage{algpseudocode}
\usepackage{amsthm}
\newtheorem{theorem}{Theorem}  
\newtheorem{definition}{Definition}
\usepackage{algorithm}
\usepackage{amsmath}

\usepackage{physics}
\usepackage{hyperref}

\begin{document}

\title{A Plug-and-Play Multi-Criteria Guidance for \\ Diverse In-Betweening Human Motion Generation}




\author{Hua Yu, Jiao Liu, Xu Gui, Melvin Wong, Yaqing Hou,
        and Yew-Soon Ong,~\IEEEmembership{Fellow,~IEEE}%
\thanks{Manuscript received xxxx.}
\thanks{Hua Yu, Jiao Liu and Melvin Wong are with College of Computing and Data Science, Nanyang Technological University, Singapore (e-mail: yu\_hua@ntu.edu.sg, jiao.liu@ntu.edu.sg, WONG1357@e.ntu.edu.sg).}
\thanks{Xu Gui and Yaqing Hou are with the School of Computer Science and Technology, Dalian University of Technology, Dalian (e-mail: houyq@dlut.edu.cn, guixuysl@gmail.com)}
\thanks{Yew-Soon Ong is with College of Computing and Data Science, Nanyang Technological University, Singapore; Centre for Frontier AI Research, Institute of High Performance Computing, Agency for Science, Technology and Research, Singapore (e-mail: asysong@ntu.edu.sg) }
}


\maketitle

\begin{abstract}
In-betweening human motion generation aims to synthesize intermediate motions that transition between user-specified keyframes. In addition to maintaining smooth transitions, a crucial requirement of this task is to generate diverse motion sequences.
It is still challenging to maintain diversity—particularly when it is necessary for the motions within a generated batch sampling to differ meaningfully from one another due to complex motion dynamics. 
In this paper, we propose a novel method, termed the \textit{Multi-Criteria Guidance with In-Betweening Motion Model} (MCG-IMM), for in-betweening human motion generation. A key strength of MCG-IMM lies in its plug-and-play nature: it enhances the diversity of motions generated by pretrained models without introducing additional parameters. 
This is achieved by providing a sampling process of pretrained generative models with multi-criteria guidance. Specifically, MCG-IMM reformulates the sampling process of pretrained generative model as a multi-criteria optimization problem, and introduces an optimization process to explore motion sequences that satisfy multiple criteria, e.g., diversity and smoothness.
Moreover, our proposed plug-and-play multi-criteria guidance is compatible with different families of generative models, including denoised diffusion probabilistic models, variational autoencoders, and generative adversarial networks.
Experiments on four popular human motion datasets demonstrate that MCG-IMM consistently state-of-the-art methods in in-betweening motion generation task. 
\end{abstract}

\begin{IEEEkeywords}
In-betweening human motion generation, Multi-Criteria Guidance, Plug-and-play nature, Generative optimization process.
\end{IEEEkeywords}

\section{Introduction}
Generative models \cite{ho2020denoising,shafir2024human}, such as Variational Autoencoders (VAEs) \cite{kingma2019introduction}, Generative Adversarial Networks (GANs) \cite{goodfellow2020generative}, and Diffusion models \cite{song2020score,ho2020denoising}, have recently emerged as a promising technique in many tasks, e.g., computer animation \cite{tmm2,tmm6}, virtual reality \cite{tmm3,tmm5}, and human-computer interactions \cite{tmm1,tmm4}. As a classical task in animation and robotics, in-betweening human motion generation has been widely studied and addressed by using generative models. Given user-provided motion sequences, these models can interpolate smooth and diverse transition motions between the input sequences. This capability is particularly valuable for producing continuous, realistic animations and reducing the need for manual keyframe design, thereby saving time and effort in motion synthesis. The need to generate diverse outputs is particularly important in real-world scenarios. In the context of in-betweening human motion generation tasks \cite{tevet2023human}, producing varied human motions conditioned on user-provided sequences enables engineers to access a broad spectrum of transitional movements. This diversity allows users to select results that best align with specific application requirements, which is critical for fields such as AI-driven animation \cite{tmm7,tmm8,tmm9}.

\begin{figure}[tb]
    \centering
    \includegraphics[width=1\linewidth]{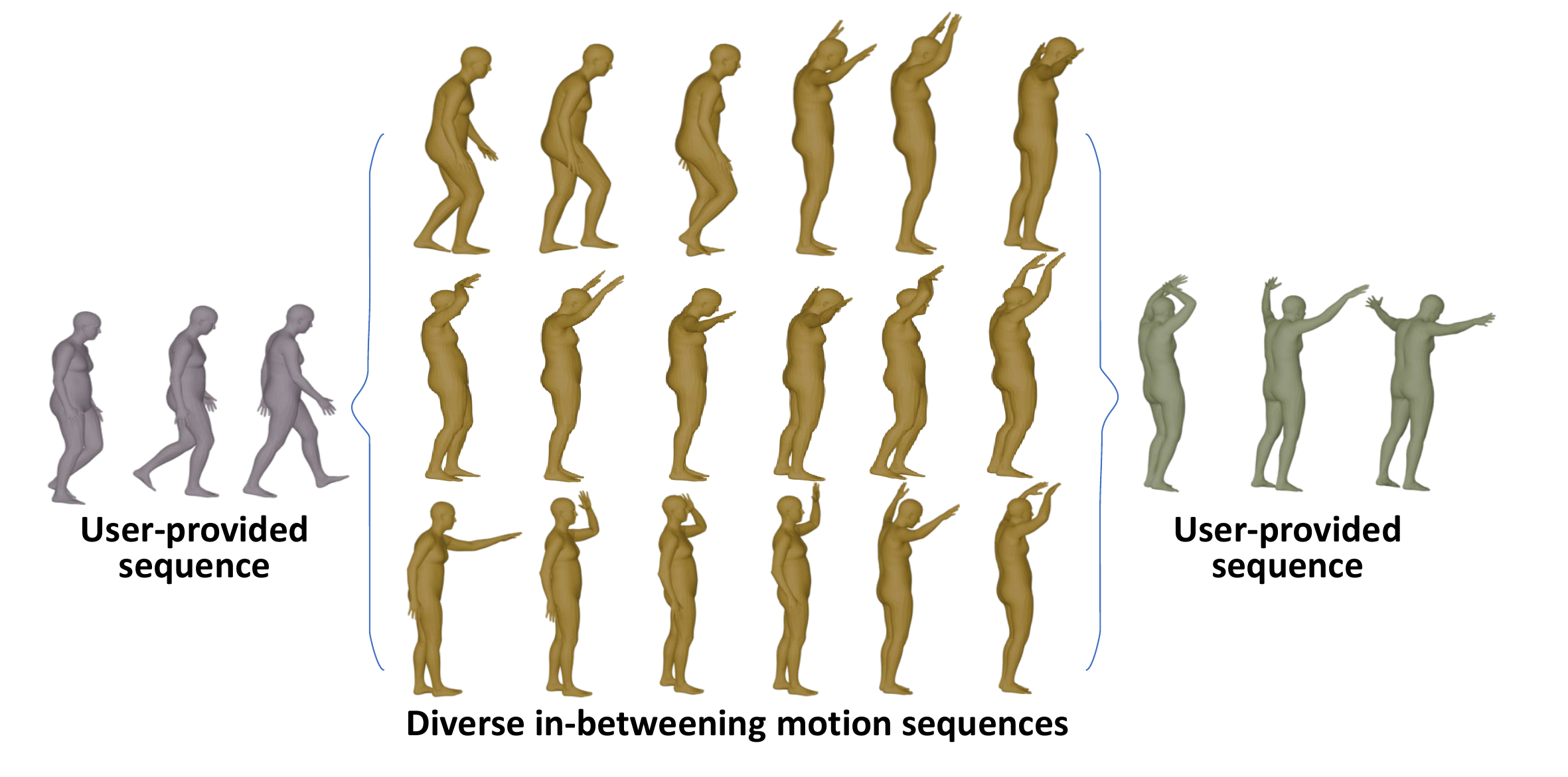}
    \caption{The example of our method MCG-IMM. Given two human motion sequences, the proposed method MCG-IMM can achieve in-betweening human motion generation with multi-criteria guidance, e.g., diversity and smoothness.}
    \label{fig:first}
\end{figure}

Although various generative models have been applied to this task \cite{vaswani2017attention,harvey2020robust}, it is still challenging for generative models to yield motion sequences with sufficient semantic diversity that satisfy users' requirements in one sampling pass.
For example, Wei \textit{et al.} \cite{mao2022weakly} utilized a VAE-based approach combined with Transformer \cite{vaswani2017attention} network to fill long-term missing motion frames. 
Zhou \textit{et al.} \cite{harvey2020robust} proposed a conditional GAN network to learn in-betweening human motions.
Diffusion methods, such as score-based generative models (SGMs) \cite{song_score-based_2021,song2020improved} and denoising diffusion probabilistic models (DDPM) \cite{ho2020denoising}, demonstrate exceptional generative performance than VAE and GAN based methods due to its stochastic processes \cite{tevet2023human, shafir2024human}. 
%

Despite the impressive performance of DDPMs, they suffer from reduced diversity under two key conditions. First, quality-oriented sampling acceleration methods—such as deterministic samplers like DDIM \cite{song2021denoising}—suppress the stochasticity essential for output variability, resulting in notable declines in motion diversity compared to vanilla DDPMs, especially on domain-restricted datasets \cite{kwon2022diffusion}. Second, strong conditional constraints further shrink the manifold of the learned distribution, causing outputs to cluster around dominant modes; this leads to a trade-off between condition accuracy and sample diversity, as evidenced by up to 22\% loss in kinematic variation in conditional settings \cite{tevet2023human,jiang2023motiondiffuser}.

Recently, enhancing the diversity of generative models for human motion in-betweening has been widely studied \cite{guo2020action2motion,mao2022weakly}. However, existing approaches primarily focus on architectural innovations—such as increasing model capacity \cite{yuan2020dlow,tmm8} or employing data-driven training strategies \cite{MotionInfilling,tmm10}—to expand the range of plausible motion sequences. These methods often come at the cost of increased training complexity, requiring task-specific modules or additional trainable parameters, thereby raising computational overhead. Moreover, they overlook a critical practical constraint: single batch diversity collapses under fixed model parameters, which means the generated motions lack dynamic diversity in a single batch sampling.
These challenges give rise to the underexplored question: whether it is possible to maximize diversity of motion sequences generated through single batch sampling, given a pretrained model and various conditions, without altering the backbone or introducing additional training parameters. 


In this paper, we design a plug-and-play \textit{Multi-Criteria Guidance with In-Betweening Motion Model} (MCG-IMM) for diverse in-betweening human motion generations. The proposed MCG-IMM framework addresses the challenge of fully unlocking the generative potential of a pretrained model to sample batches of diverse human motions. MCG-IMM operates at the sampling stage of a pretrained model, requiring no additional training. This also makes MCG-IMM a plug-and-play framework that can be seamlessly integrated with a wide range of generative models—including DDPMs, VAEs, and GANs—to effectively enhance the diversity of the generated motion sequences.
Unlike prior approaches that require fine-tuning or model-specific architectural changes \cite{kwon2022diffusion,jiang2023motiondiffuser}, our plug-and-play design enables seamless integration with any off-the-shelf generative backbone, significantly reducing development time and computational overhead.
It is demonstrated that, without introducing any additional training parameters, the proposed MCG-IMM effectively guides the generation of diverse in-betweening human motions. Moreover, as a plug-and-play framework, the proposed method can effectively enhance the diversity of various generative models, including DDPM, VAE, and GAN. 
The main contributions of this paper are summarized as follows:

\begin{itemize}
\item We propose MCG-IMM, a plug-and-play multi-criteria guidance for in-betweening human motion generation. 
The key advantage of MCG-IMM lies in its plug-and-play nature—it can be seamlessly integrated with various generative backbone models to produce diverse and smooth human motions without requiring retraining or parameter adjustment.
\item We introduce two carefully designed criteria in the multi-criteria optimization process to explore a broad range of realistic and coherent motion samples during the single batch sampling process. 
\item Systematic experiments are conducted to validate the effectiveness of the proposed MCG-IMM framework. The results demonstrate that, regardless of the underlying generative architecture, the proposed approach effectively guides the generation of in-betweening motion sequences that are both diverse and smooth.
\end{itemize}


The organization of the rest of this paper is as follows.
Section \ref{Related Work} presents the existing works related to our task. 
Section~\ref{MOP} presents the details of the proposed method, including the underlying principles of the designed criteria and the EA based sampling process.
Section \ref{ExpDesign} gives a description of the experimental design. Section \ref{Results} analyzes the results of the extensive experiment, including the quantitative and qualitative results. Finally, Section \ref{conclusion} concludes the paper.

\section{Related Work}\label{Related Work}
Diverse in-betweening human motion task aims to generate diverse and smooth motion sequences given user-provided sequences. 
A plethora of generative models, such as VAEs and GANs, have been applied to this task \cite{mao2022weakly, fragkiadaki2015recurrent,li2023sequential,agrawal2013diverse}. 
For example, 
Harvey \textit{et al.} \cite{harvey2020robust} utilized neural networks to generate plausible interpolation human motions between a given pair of keyframe poses. 
Tang \textit{et al.} \cite{tang2022real} built upon these findings by introducing a Convolutional Variational Autoencoder (CVAE) \cite{mao2022weakly}, which utilizes the motion manifold and conditional transitioning to generate real-time motion transitions.
Motion DNA \cite{zhou2020generative} proposed to automatically synthesize complex motions over a long time interval given very sparse keyframes by users for long-term in-betweening.
Li \textit{et al.} \cite{li2023sequential} transferred the end pose of the previous motion to the next for a cohesive transition. 
However, VAEs assume a Gaussian distribution as the posterior, which can limit the diversity of the generated samples \cite{10.1145}. Meanwhile, GANs tend to mainly generate samples from the major modes while ignoring the minor modes, further restricting the overall diversity.
More recently, denoising diffusion models \cite{ho2020denoising, lee2023multiact, guo2020action2motion,xie2024omnicontrolcontroljointtime} have been extensively utilized for motion generation. MoFusion \cite{dabral2023mofusion} proposed a denoising diffusion-based framework to generate motion in-betweening by fixing a set of keyframes in the motion sequence and reverse-diffusing the remaining frames.
OmniControl \cite{xie2024omnicontrolcontroljointtime} presents a novel method for text-conditioned human motion generation, which incorporates flexible spatial control signals over different joints at different times. Their diffusion backbone is based on the human motion diffusion method. 
Based on this, CondMDI \cite{10.1145/3641519.3657414} proposes a flexible in-betweening pipeline through random sampling keyframes and concatenating binary masking during training. 

Although various methods have been proposed to enhance diversity in in-betweening motion generation, most focus on modifying model architectures \cite{10814072,Hua_2025_CVPR} or improving training strategies—often at the cost of added parameters and increased complexity. Moreover, even with diverse sampling capabilities, models still struggle with generating distinct motions within a single batch, a challenge of intra-batch diversity that is often overlooked. Unlike prior works \cite{tan2025fast,wei2025evolvable}, this paper shifts the focus to the sampling process, introducing a plug-and-play guidance module that promotes intra-batch diversity without requiring additional training or parameter tuning. The method is lightweight and broadly compatible with existing generative models.

\section{Proposed Method}\label{MOP}

In this section, we present the proposed plug-and-play multi-criteria guidance designed to generate diverse human motions. We begin by introducing two criteria specifically designed to guide the sampling process, grounded in the principles of multi-criteria optimization. Next, we describe how these criteria are employed to steer the optimization process in exploring diverse motion samples based on a pretrained generative model. We then use the DDPM as an example to illustrate the acquisition and use of a pretrained backbone. Finally, we summarize the overall framework of the proposed method.

\subsection{ Multi-Criteria for Human Motion In-Betweening Task}
\begin{figure*}
    \centering
    \includegraphics[width=1\linewidth]{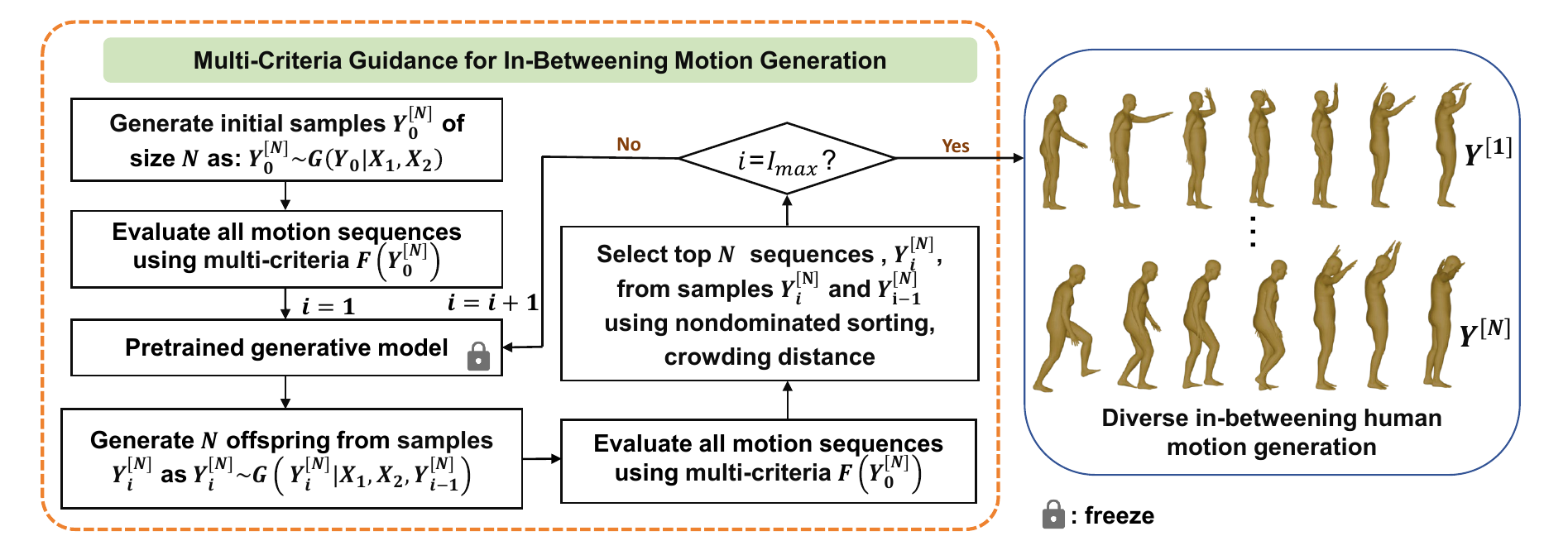}
    \caption{The illustration of multi-criteria guidance for diverse in-betweening human motion sequences. $Y$ is the generated sequence. $G$ denotes the pretrained generator of the generative model. $F$ denotes the designed multi-criteria function. }
    \label{fig:moea}
\end{figure*}

In this section, we introduce the two formulated criteria used to guide the sampling process. These criteria comprise two components: the \textit{Diversity Component}, which encourages variation among generated motions, and the \textit{Smoothness Component}, which ensures temporal coherence and realism.

\begin{itemize}
    \begin{figure}[!h]
        \centering
        \includegraphics[width=1\linewidth]{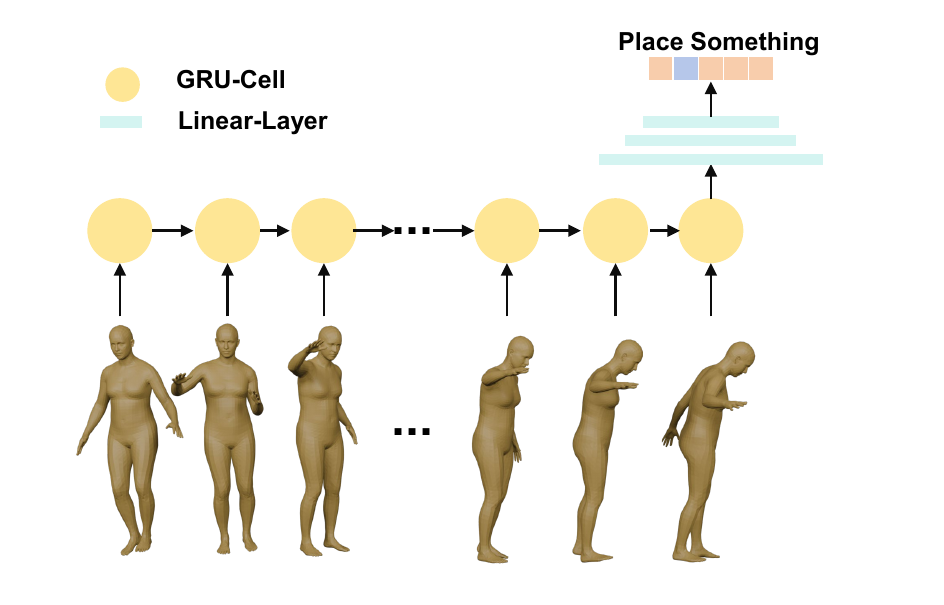}
        \caption{The illustration of a multi-class classifier model for the Diversity Component in the multi-criteria. }
        \label{fig:cls}
    \end{figure}
    
    \item \textbf{\textit{Diversity Component}}: {
    The diversity of human motions is reflected in the diversity of action categories and the diversity of intra-class differences within the same human action type.
    We assume the availability of a classifier $C({Y})$ that can categorize the motion type represented by a generated motion sequence $Y$. In Fig. \ref{fig:cls}, we provide the network structure of the multi-class classifier model used for the \textit{Diversity Component}. This model consists of GRU layers~\cite{69e088c8129341ac89810907fe6b1bfe} to encode the temporal information and MLP layers to produce the final classification results. 
   The classifier is assumed to be capable of distinguishing $D$ types of motions, each labeled by an integer in the set $\{ 0,1,\cdots, D-1 \}$.} Thus, for any sequence $Y$, we have $C({Y}) \in \{ 0,1,\cdots,D-1 \}$. 
   Thus, $C({Y})$ can distinguish between different categories and encourage the generation of samples across action categories.
   In addition, the classifier can provide the probability $P_{c}(Y)$ that a given motion sequence belongs to each categorized motion type, which denotes different probabilities corresponding to the same action label (intra-class difference within the same human action type).
   $P_{c}(Y)$ reflects the changes in internal samples of the category, and helps to avoid generating single poses in the same category.
   Based on the above definition, the diversity component is formulated as the following two functions $\alpha_1(Y)$ and $\alpha_2(Y)$:
    \begin{equation}\label{EqMulObj1}
    \begin{aligned}
    \begin{cases}
    \alpha_1 (Y) = \frac{1}{D} \left(C(Y) + P_{c}(Y) \right), \\
    \alpha_2 (Y) = 1- \alpha_1 (Y).
    \end{cases}
    \end{aligned}
    \end{equation}

    \item \textbf{\textit{Smoothness Component}}: 
    {A smoothness function $\beta (Y)$ is utilized to facilitate the smoothness of adjacent motion sequences, which is formulated as follows:
    \begin{equation}\label{eqn:smooth_comp}
        \beta(Y) = \|X_1[-1]-Y[0]\| + \|Y[-1]-X_2[0]\|.
    \end{equation}
    In \eqref{eqn:smooth_comp}, $X_1[-1]$ and $Y[-1]$ denote the last human pose of the user-provided motion sequence $X_1$ and the in-betweening sequence $Y$, respectively. $X_2[0]$ and $Y[0]$ denote the first human pose of the motion sequence $X_2$ and the generated in-betweening sequence $Y$, respectively.
    $\beta(Y)$ is utilized to denote the offset between the generated in-betweening motion sequences $Y$ and the adjacent motion sequences ($X_1$ and $X_2$). This operation aims to guarantee that the changes between the adjacent human motion poses are not too large, thereby enhancing the smoothness between different action sequences.}
\end{itemize}

By integrating the diversity and smoothness components, a \textit{multi-criteria optimization problem} for the diverse human motion in-betweening task can be formulated as follows:
\begin{align}
\begin{aligned}\label{EqMulObj2}
   \begin{cases}
    \min : F_1 (Y)= \alpha_1 (Y) + \beta (Y), \\
    \min : F_2 (Y) = \alpha_2 (Y) + \beta (Y).
   \end{cases}
\end{aligned}
\end{align}
The multi-criteria optimization problem is a classical concept in optimization that involves simultaneously optimizing multiple, often conflicting, objectives~\cite{8476217,6766232,10195116}. Several fundamental definitions associated with multi-criteria optimization are outlined below.
\begin{definition}[Pareto Dominance]
For decision vectors $\textbf{x}_a$ and $\textbf{x}_b$, if $\forall i \in \{1,2,\cdots,m\}$, $f_{i}(\textbf{x}_a) \leq f_{i}(\textbf{x}_b)$ and $\exists j \in \{1,2,\cdots,m\}$, $f_{j}(\textbf{x}_a) < f_{j}(\textbf{x}_b)$, $\textbf{x}_a$ is said to {Pareto dominate} $\textbf{x}_b$.
\end{definition}

\begin{definition}[Pareto-Optimal Solution]
If no decision vector in $\mathcal{X}$ Pareto dominates $\textbf{x}_a$, then $\textbf{x}_a$ is a {Pareto-optimal} solution.
\end{definition}
Given the above definitions, the \textit{Pareto set} refers to the set of all Pareto-optimal solutions in the decision space, while the \textit{Pareto front} denotes the image of the Pareto set in the criteria space. Subsequently, we introduce two theorems that formally summarize the key properties of the constructed multi-criteria optimization problem.


\begin{theorem}
Let $\mathcal{B} = \{ Y_b| Y_b = \arg \min_Y \beta (Y)\}$. Assuming $|\mathcal{B}| \geq 2$, we can state that any in-betweening motion $Y_b \in \mathcal{B}$ is a Pareto optimal solution of the multi-criteria optimization problem defined in \eqref{EqMulObj2}.
\end{theorem}

\begin{proof}
Assuming that there exist an in-betweening motion sequence $Y_b^{*}$ that can Pareto dominate $Y_b$ in problem \eqref{EqMulObj2}, then we have:
\begin{align}
    \begin{aligned}
        F_1 (Y_b^{*}) <  F_1 (Y_b) , \\
        F_2 (Y_b^{*}) <  F_2 (Y_b) .
    \end{aligned}
\end{align}
Thus, we can obtain that 
\begin{align}
    \begin{aligned}
    \label{proof1}
        & F_1 (Y_b^{*}) + F_2 (Y_b^{*}) <  F_1 (Y_b) + F_2 (Y_b) \\
        \Rightarrow & \alpha_1 (Y_b^{*}) + \alpha_2 (Y_b^{*}) + 2\beta (Y_b^{*}) < \alpha_1 (Y_b) + \alpha_2 (Y_b) + 2\beta (Y_b) \\
        \Rightarrow & 1 + 2\beta (Y_b^{*}) < 1 + 2\beta (Y_b) \\
        \Rightarrow & \beta (Y_b^{*}) < \beta (Y_b),
    \end{aligned}
\end{align}
where line 2 and line 3 in \eqref{proof1} are obtained based on the \eqref{EqMulObj1}. It should be noticed that, we have $\beta (Y_b^{*}) < \beta (Y_b)$ from \eqref{proof1}, and this contradicts the assumption that $Y_b \in \mathcal{B}$. Therefore, there is no in-betweening human motion sequence that can Pareto dominate $Y_b$, which demonstrates our proof.
\end{proof}

\textit{Theorem 1} establishes that all in-betweening human motion sequences that minimize the smoothness component $\beta(\cdot)$ are Pareto optimal solutions for the multi-criteria optimization problem defined in \eqref{EqMulObj2}. This result implies that solving \eqref{EqMulObj2} to identify its Pareto optimal motions inherently favors those with low, or even minimum, values of $\beta(\cdot)$. Given that $\beta(\cdot)$ governs the smoothness of transitions between an initial human pose and subsequent in-betweening motions, the Pareto optimal solutions are expected to facilitate smooth motion transitions, thereby ensuring that the generated motions exhibit continuity and smoothness. 

\begin{theorem} 
Let $\mathcal{B} = \{ Y_b| Y_b = \arg \min_Y \beta (Y)\}$, and $Y_1, Y_2 \in \mathcal{B}$. If $||\textbf{F}(Y_1),\textbf{F}(Y_2)|| > \frac{4}{D}$, where $\textbf{F}(\cdot) = ({F}_1(\cdot),{F}_2(\cdot))$ and $||\cdot||$ is the Manhattan distance, then $C(Y_1) \neq C(Y_2)$.
\end{theorem}

\begin{proof}
Considering that $Y_1, Y_2 \in \mathcal{B}$, we have $\beta(Y_1) = \beta(Y_2) = \min_Y \beta(Y) $. According to \eqref{EqMulObj2}, we have the following equations: 
\begin{align}
\begin{aligned}
    & {F}_1(Y_1)-{F}_1(Y_2) 
    \\ &= \frac1D\left({C(Y_1)}+P_c(Y_1)-C(Y_2)-P_c(Y_2)\right)  \\
    & F_2(Y_1)-F_2(Y_2)  
    \\&=  -\frac1D(C(Y_1)+P_c(Y_1)) +\frac1D(C(Y_2)+P_c(Y_2))
\end{aligned}
\end{align}
Then, we can obtain that  
\begin{align}
\begin{aligned}
   &  ||\mathbf{F}(Y_1)-\mathbf{F}(Y_2)|| \\
   &= |{F}_1(Y_1)-{F}_1(Y_2) | + | F_2(Y_1)-F_2(Y_2) | >\frac4D \\
  \Rightarrow & \left|\frac1D\left(C(Y_1)+P_c(Y_1)-C(Y_2)-P_c(Y_2)\right)\right| \\ + 
  & \left|\frac1D (C(Y_2)+P_c(Y_2)- C(Y_1)-P_c(Y_1))\bigg|>\frac4D \right.      \\
  \Rightarrow & 2\left|\frac1D\left(C(Y_1)+P_c(Y_1)-C(Y_2)-P_c(Y_2)\right)\right|>\frac4D \\
  \Rightarrow &  |C(Y_1)+P_c(Y_1)- C(Y_2)-P_c(Y_2)|> 2 \\
  \Rightarrow &  |C(Y_1) - C(Y_2)| > 2 - |P_c(Y_1) -P_c(Y_2)|
\end{aligned}
\end{align}

Considering that both $P_c(Y_1)$ and $P_c(Y_2)$ are probabilities, i.e., $P_c(Y_1) \in [0,1]$ and $P_c(Y_2) \in [0,1]$, thus we have $|P_c(Y_1)-P_c(Y_2)| \leq 1$, thus resulting in $|C(Y_1) - C(Y_2)| > 1$. Given that $C(Y)$ is the result of the classifier $C$ applied to sequence $Y$ and that its value is an integer, we can obtain the conclusion that $|C(Y_1)-C(Y_2)|>1$ is equivalent to $C(Y_1) \neq C(Y_2)$. This shows that sequence $Y_1$ and sequence $Y_2$ correspond to different categories of motions.
\end{proof}

\textit{Theorem 2} demonstrates that, under the condition of the smooth transition, if the Manhattan distance between two generated in-betweening motions in the objective space exceeds $\frac4D$, the human motion sequences are likely to belong to different motion categories. This finding suggests that a diverse set of Pareto optimal motions from \eqref{EqMulObj2} not only maintains smooth transitions but also provides a variety of motion categories. The multi-criteria optimization module is particularly suited to achieving this goal, offering a robust mechanism for generating smooth and diverse in-betweening motions.

\textbf{\textit{Remark 1}}: \textit{Although we have presented two theorems to explain the properties of the constructed multi-criteria optimization problem, some readers may still find it challenging to grasp. To enhance clarity, we provide an intuitive explanation in the appendix to illustrate the underlying principles of the constructed problem. For further details, please refer to Section B of the appendix.}


\subsection{Generative Optimization of Diverse In-Betweening Human Motions}
In the proposed MCG-IMM framework, we propose a generative optimization approach to solve the formulated multi-criteria optimization problem, preserving diversity within the criteria space.
The designed MCG-IMM is illustrated in Fig. \ref{fig:moea}. In general, the details for the generative optimization approach is as follows:
\begin{enumerate}
    \item \textbf{\textit{Initialization}}: The pretrained generative model is employed to produce initial samples of human motion sequences. These sequences are subsequently evaluated using the criteria defined in equation~\eqref{EqMulObj2}.

    \item \textbf{\textit{Offspring Generation}}: Generating a set of offspring human motion sequences using the pretrained generative model.

    \item \textbf{\textit{Evaluation}}: Evaluating the offspring using the criteria defined in equation~\eqref{EqMulObj2}.

    \item \textbf{\textit{Selection}}: Performing fast nondominated sorting and calculating the crowding distance for the sequences in both the current samples and the offspring samples. Subsequently, the samples will be updated through the elite selection.

    \item Repeating steps 2 to 4 until the termination condition is met.
\end{enumerate}

\begin{figure*}[!htb]
    \centering
    \includegraphics[width=0.9\linewidth]{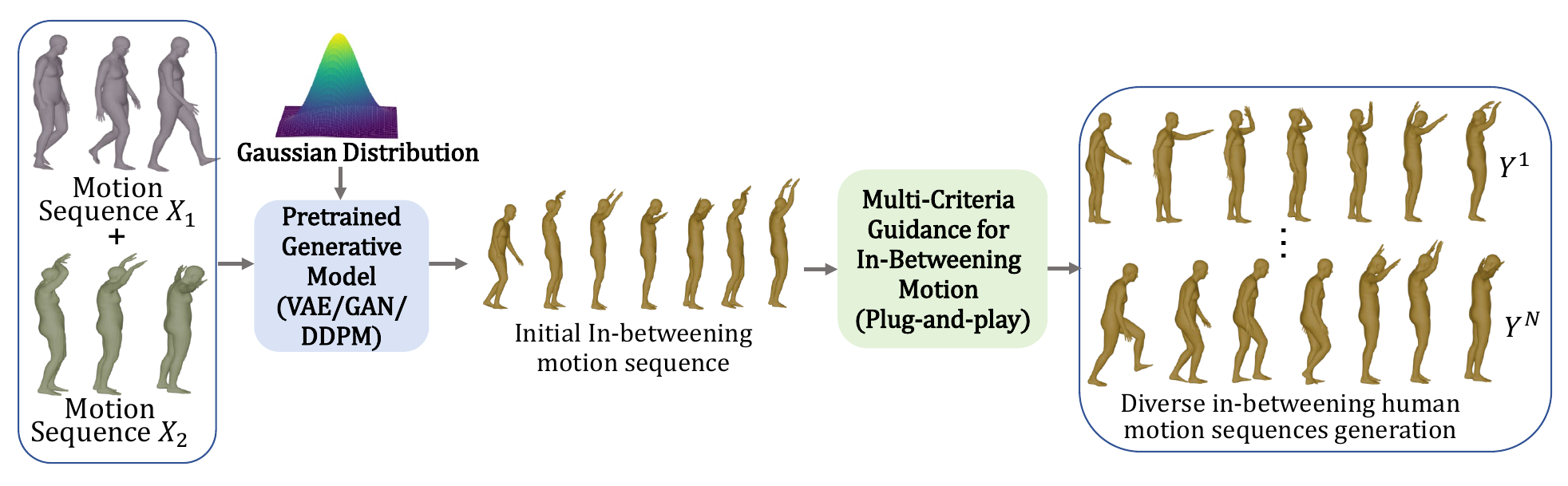}
    \caption{The framework of MCG-IMM. We first randomly sample from the Gaussian distribution, and use the pretrained generator to generate the initial in-betweening motion sequence. Then, the designed multi-criteria generation framework is utilized to explore multiple optimal solutions for more diverse sequences. }
    \label{fig:overview}
\end{figure*}

The uniqueness of the proposed method lies in the incorporation of a generative model into the optimization process. Specifically, the generation process is formulated as follows:
\begin{equation}\label{eqn:ea}
    \left.{Y}_i^{[n]}\sim\left\{\begin{array}{ll}G\left(Y_i\mid X_1, X_2\right) & i=0 \\
    G\left(Y_i\mid Y_{i-1}^{[n]}, X_1, X_2 \right) & 1\leq i \le I \end{array}\right.\right. ,   \\
\end{equation}
where $G$ is the generative model, $Y_{i}^{[n]}$ ($n \in \{ 1,...,N \}$) is the $n$th human motion sequence in the samples at the $i$th iteration. It can be observed that in the first iteration, i.e., during the initialization process, the human motion sequence is generated by the conditional generative model on the user-provided sequences $X_1$ and $X_2$. In subsequent iterations, i.e., during the evolution process, the generative model is conditioned not only by the user-provided sequences but also by the sequences already generated and present in the current samples. This approach allows the generative model to search for Pareto-optimal sequences based on the elite sequences identified so far, thereby driving the evolution of the samples. Moreover, the crowding distance also contributes to maintaining the criteria space diversity of the generated human motions. After achieving the max iterations, we can get $N$ optimal solutions, i.e., diverse and smooth in-betweening human motion sequences. 

Moreover, to allow for more flexible in-betweening human motion generation, the length of the generated motion sequence is variable, as different action sequences may require different transition lengths. The desired motion length is estimated according to the following steps: (1) calculating the cosine similarity $S \in [0,1]$ to measure the similarity between two given motion poses $X_1 [-1]$ and $X_2 [0]$, and (2) based on the preset minimum sequence length $Y_{min}$ and the preset maximum sequence length $Y_{max}$, the desired length of the human motion sequence is set as follows:
\begin{equation}
    Y_{len}= Y_{min} + \left \lfloor(Y_{max}- Y_{min}) \times (1-S) \right \rfloor.
\end{equation}
In our experiment, $Y_{min}$ is set to 5, $Y_{max}$ is set to 15. 
In the generation of the offspring, the $Y_{len}$ is encoded into the pretrained generative model and has to be consistent with the training process details. So we introduce a padding operation to the $Y_{len}$, specifically, we repeat the last pose of the $Y_{len}$ until it achieves the $Y_{max}$. 
Note that the proposed structure is readily adaptable to existing Conditional generative models, such as VAE, GAN and DDPM, by aligning their loss functions and training protocols with standard implementations. Detailed settings using MCG-IMM with the pretrained conditional VAE, GAN, and DDPM generation model are provided in Appendix.

\begin{algorithm}[t]
    \caption{In-Betweening Motion Generation with Multi-Criteria Guidance}
    \label{alg1}
    \begin{algorithmic}[1]
        \Require {Generative model $G$, user-provided motion sequences $X_1$ and $X_2$, number of generations $\tau_{max}$, number of samples $l$, number of offspring $m$ in each iteration.}
        \Ensure {In-betweening human motion sequences $Y^{[l]}$}
        \State Generate initial human motion sequences $Y_{i=0}^{[l]}$ via the generator, i.e., $Y_{i=0}^{[l]}\sim G(Y_{i=0}|X_1, X_2), (l \in \{1,...,N\})$
        \State Evaluate $Y_{i=0}^{[l]}$ using the designed criteria $F_1(Y_{i=0}^{[l]})$ and $F_2(Y_{i=0}^{[l]})$;
        \While {$i\le \tau_{max}$}
        \State Generate $m$ new motion sequences by the generator based on randomly selected human motion sequences $[l']$ in the samples, i.e., $Y_{i}^{[m]}\sim G(Y_{i}^{[m]}|Y_{i-1}^{[l']}, X_1, X_2)$;
        \State Evaluate the new motion sequences $Y_{i}^{[m]}$ using the designed multiple criteria $F_1(Y_{i}^{[m]})$ and $F_2(Y_{i}^{[m]})$;
        \State Select top $l$ motion sequences from $Y_{{i}}^{[l]} \cup Y_{{i}}^{[m]}$ using nondominated rank and crowding distance, thus forming the new samples $Y_{{i+1}}^{[l]}$
        \EndWhile \\
        \Return $Y^{[l]}_i$
    \end{algorithmic}
\end{algorithm}

\subsection{A Summary for the Proposed Method}
In summary, our task aims to generate diverse and smooth in-betweening human motions given the user-provided human motion frames. 
The overflow of the proposed method is briefly described in Fig. \ref{fig:overview}.
Specifically, given the user-provided motion sequence $X_1$ and $X_2$, we first generate the initial in-betweening human motions $Y$ through a designed and pretrained generative model (such as CVAE, GAN and DDPM).
To enhance the diversity and smoothness of the generated motion sequence, we transform this task into a multi-criteria optimization problem, and design a plug-and-play module to guide the in-betweening motion generation. We have demonstrated that for any Pareto optimal solutions of the formulated problem, the corresponding human motion sequences can support a smooth and natural transition to interpolate user-provided sequences. Therefore, the multi-criteria optimization process can capture diverse optimal solutions, resulting in the generation of diverse in-betweening human motion sequences.
It is significant to note that the proposed method can enhance the diversity of generated human motions based on generative models without introducing additional training processes and parameters.

\begin{table*}[!t]
\centering
\caption{Quantitative comparison results. $FID_{tr}$ and $FID_{te}$ refer to the FID scores obtained from the generation to train and test datasets, respectively. The best results are in bold.}
\label{tab:sota}
\resizebox{16.5cm}{!}{
\begin{tabular}{c|ccccc|ccccc}
\hline
\multirow{2}{*}{Method}                               & \multicolumn{5}{c|}{BABEL}                                                                      & \multicolumn{5}{c}{HAct12}                                                                    \\ \cline{2-11} 
                                                      & $FID_{tr}\downarrow$ & $FID_{te}\downarrow$ & $ACC\uparrow$  & $ADE\downarrow$ & $APD\uparrow$ & $FID_{tr}\downarrow$ & $FID_{te}\downarrow$ & $ACC\uparrow$ & $ADE\downarrow$ & $APD\uparrow$ \\ \hline
RMI \cite{harvey2020robust}          & 37.09                & 30.15                & 1.51           & 1.21            &

0.79          & 245.35               & 298.06               & 24.51          & 1.38            & 0.60          \\
MITT \cite{qin2022motion}            & 32.21                & 27.10                & 0.73           & 0.99            & 0.91          & 254.72               & 143.71               & 22.73          & 1.39            & 0.53          \\
Motion DNA \cite{zhou2020generative} & 27.04                & 23.25                & 16.2           & 1.12            & 0.67          & 247.68               & 139.89               & 24.46          & 1.34            & 0.87          \\
ACTOR \cite{petrovich2021action}     & 29.34                & 30.31                & 40.9           & 2.29            & 2.71          & 248.81               & 381.56               & 44.41          & 1.54            & 0.95          \\
WAT (RNN) \cite{mao2022weakly}       & 22.54                & 22.39                & 49.6           & 1.47            & 1.74          & 129.95               & 164.38               & 59.02          & 1.23            & 0.96          \\
WAT (Trans.) \cite{mao2022weakly}    & 20.02                & 19.41                & 39.5           & 1.40            & 1.82          & 141.85               & 139.82               & 56.87          & 1.26            & 0.88          \\
MultiAct \cite{lee2023multiact}      & 16.39                & 19.12                & 73.67          & 1.77            & 5.42          & 174.76               & 243.82               & 68.62         & 1.35            & 1.76          \\
MoFusion \cite{dabral2023mofusion}   & 15.49                & 15.12                & 74.71          & 1.07            & 6.45          & 125.41               & 134.14               & 67.14          & 1.07            & 1.45          \\ \hline
MCG-IMM (VAE)                                         & \textbf{14.13}       & 15.71                & 74.21          & 1.12            & 6.14          & \textbf{115.36}      & 135.26               & 67.27          & 1.21            & 1.81          \\
MCG-IMM (GAN)                                         & 14.46                & 14.49                & 75.16          & 1.10            & 6.96          & 116.36               & 130.49               & 68.91          & 1.20            & 1.92          \\
MCG-IMM (DDPM)                                        & 14.32                & \textbf{13.65}       & \textbf{77.32} & \textbf{1.01}   & \textbf{8.45} & 116.16               & \textbf{121.34}      & \textbf{70.26} & \textbf{1.17}   & \textbf{2.45} \\ 
\hline
\hline
\multirow{2}{*}{Method}                               & \multicolumn{5}{c|}{NTU}                                                                        & \multicolumn{5}{c}{GRAB}                                                                      \\ \cline{2-11} 
                                                      & $FID_{tr}\downarrow$ & $FID_{te}\downarrow$ & $ACC\uparrow$  & $ADE\downarrow$ & $APD\uparrow$ & $FID_{tr}\downarrow$ & $FID_{te}\downarrow$ & $ACC\uparrow$ & $ADE\downarrow$ & $APD\uparrow$ \\ \hline
RMI \cite{harvey2020robust}          & 144.98               & 113.61               & 66.3           & 1.11            & 1.19          & 132.28               & 102.54               & 75.1          & 1.24            & 1.79          \\
MITT \cite{qin2022motion}            & 151.11               & 157.54               & 70.6           & 1.20            & 1.21          & 162.34               & 148.32               & 62.1          & 1.21            & 1.19          \\
Motion DNA \cite{zhou2020generative} & 147.64               & 147.92               & 68.7           & 1.19            & 1.09          & 157.99               & 146.18               & 60.7          & 1.08            & 0.97          \\
ACTOR \cite{petrovich2021action}     & 355.69               & 193.58               & 66.3           & 1.49            & 2.07          & 342.73               & 181.45               & 78.1          & 1.45            & 2.13          \\
WAT (RNN) \cite{mao2022weakly}       & 72.18                & 111.01               & 76.0           & 1.20            & 2.20          & 70.12                & 101.24               & 76.0          & 1.18            & 2.05          \\
WAT (Trans.) \cite{mao2022weakly}    & 83.14                & 114.62               & 71.3           & 1.23            & 2.19          & 79.17                & 114.52               & 69.1          & 1.23            & 2.19          \\
MultiAct \cite{lee2023multiact}      & 374.73               & 530.09               & 63.9           & 1.41            & 2.73          & 374.73               & 530.09               & 64.3          & 1.32            & 2.54          \\
MoFusion \cite{dabral2023mofusion}   & 75.46                & 108.12               & 77.71          & 1.06            & 2.75          & 68.12                & 95.01                & 79.21         & 1.09            & 2.01          \\ \hline
MCG-IMM (VAE)                                         & \textbf{81.64}       & \textbf{107.16}      & 74.9           & 1.09            & 2.72          & 69.46                & 98.16                & 78.5          & 1.13            & 2.37          \\
MCG-IMM (GAN)                                         & 79.47                & 107.51               & 76.4           & 1.05            & 2.86          & 67.19                & 96.49                & 80.1          & 1.08            & 2.42          \\
MCG-IMM (DDPM)                                        & \textbf{71.90}       & \textbf{107.32}      & \textbf{79.8}  & \textbf{1.03}   & \textbf{2.94} & \textbf{65.32}       & \textbf{94.72}       & \textbf{85.9} & \textbf{1.04}   & \textbf{2.59} \\ \hline
\end{tabular}
}
\end{table*}

\section{Experiments and Design}
\label{ExpDesign}

\subsection{Datasets}
The experiments are performed on four widely employed motion datasets: BABEL \cite{punnakkal2021babel}, HumanAct12 \cite{guo2020action2motion}, NTU RGB-D \cite{liu2019ntu}, GRAB \cite{taheri2020grab}.

\textbf{BABEL} is a subset of the AMASS dataset with per-frame action annotations. In this work, we split the dataset into two parts: single-action sequences and transition sequences between two actions. We downsample all motion sequences to $30$ Hz. For single-action motions, we divide the long motions into several short ones. Each short motion performs one single action, and the remove too short sequences ($< 1$ second). We also eliminate the action labels with too few samples ($< 60$) or overlap with other actions. After processing, we have $20$ action categories. 

\textbf{HumanAct12} contains $12$ subjects in which $12$ categories of actions with per-sequence annotation are provided. The sequences with less than 35 frames are removed, which results in $727$ training and $197$ testing sequences. Subjects P1 to P10 are used for training, P11 and P12 are used for testing.

\textbf{NTU RGB-D} originally contains over 100,000 motions with 120 classes whose pose annotations are from MS Kinect readout, which makes the data highly noisy and inaccurate. To facilitate training, our operation is consistent with the operation in \cite{kocabas2020vibe}, which re-estimates the 3D positions of 18 body joints (i.e. 19 bones) from the point cloud formed by aligning synchronized video feeds from multiple cameras. Note the poses are not necessarily matched perfectly to their true poses. It is sufficient here to be perceptually natural and realistic.

\textbf{GRAB} consists of 10 subjects interacting with 51 different objects, performing 29 different actions. Since, for most actions, the number of samples is too small for training, we choose the four action categories with the most motion samples, i.e., Pass, Lift, Inspect and Drink. We use 8 subjects (S1-S6, S9, S10) for training and the remaining 2 subjects (S7, S8) for testing. In all cases, we remove the global translation. The original frame rate is 120 Hz. To further enlarge the size of the dataset, we downsample the sequences to 15-30 Hz.

\subsection{Parameter Settings}
For the experiment settings, the batch size for training the conditional DDPM model is set to 128 and the number of used human joints is 16. The Transformer uses 8 attention heads. The proposed method is implemented using the PyTorch framework in Python 3.6. To ensure convergence, the Adam optimizer is used to train the model, with an initial learning rate of $10^{-2}$ that decays by 0.98 every 10 epochs. The training is conducted for 500 epochs. In this work, the parameter $\beta$ is set from 0.0001 to 0.05, where $\beta_K$ are linearly interpolated (1 $\textless$ $k$ $\textless$ $K$), which aims to ensure that the reverse process fluctuates slightly.
For the MCG-IMM, the $I_{max}$ is set to 20, the nondominated sorting and the crowding distance \cite{deb2002fast} are utilized in the environment selection. The samples size is set to 20.
All the inference processes are conducted on the NVIDIA Tesla A100 GPU.
The keyframes of the user-provided sequence are set to 5 in our experiment, and can be variable according to different tasks.

\begin{table*}[!t]
\caption{Ablation studies for the influence of different lengths on the BABEL dataset. }
\label{tab:AblationDiffLength}
\centering
\resizebox{16cm}{!}{\begin{tabular}{c|ccccc|ccccc}
\hline
              & \multicolumn{5}{c|}{ Fixed length $N = 20$}                                                               & \multicolumn{5}{c}{Variable length}                                                                 \\ \cline{2-11}
              & $FID_{tr}\downarrow$ & $FID_{te}\downarrow$ & $ACC\uparrow$ & $ADE\downarrow$ & $APD\uparrow$ & $FID_{tr}\downarrow$ & $FID_{te}\downarrow$ & $ACC\uparrow$ & $ADE\downarrow$ & $APD\uparrow$ \\
              \cline{1-11}
RMI \cite{harvey2020robust}          & 37.09                & 30.15                & 1.51         & 1.21           & 0.79         & 36.87                & 30.03                & 2.91         & 0.59           & 1.51         \\
MITT \cite{qin2022motion}         & 32.21                & 27.10                & 0.73         & 0.99           & 0.91             & 31.62                & 26.64                & 1.66         & 0.65           & 1.61             \\
CMIB \cite{kim2022conditional}         & 21.07                & 23.60                & 1.96         & 1.21           & 2.46             & 20.42                & 21.34                & 2.04         & 0.91           & 2.94             \\
MultiAct \cite{lee2023multiact}     & 16.39                & 19.12                & 73.70        & 1.77          & 3.01         & 15.98                & 13.62                & 73.99         & 0.61           & 5.42         \\
\cline{1-11}
MCG-IMM (DDPM)          &  14.91          &   13.72                   &  75.81            &  1.09              &  7.24             & 14.32 & 13.65 & 77.32 & 1.01   & 8.45           \\
\hline
\end{tabular}}
\end{table*}
\vspace{-5pt}

\subsection{Metrics}
The accuracy and diversity of the generated motion sequences are essential for in-betweening motion task. For comparison, we employ the metrics to facilitate the evaluation of MCG-IMM, i.e., \textit{Frechet Inception Distance} (FID), \textit{Action Accuracy} (ACC) and \textit{Average Displacement Error} (ADE). 
Specifically, APD aims to evaluate the diversity performance of MCG-IMM. ACC, FID and ADE aim to evaluate the accuracy performance of MCG-IMM. For APD and ACC metrics, a higher value is better. For FID and ADE metrics, a lower value is better.

\begin{enumerate}
    \item \textit{Frechet Inception Distance} (FID). FID is the distribution similarity between the predicted sequences and the ground-truth motions:
\begin{equation}
\label{equ:12}
\begin{aligned}
FID= & \|\mu_{\mathrm{gen}} - \mu_{\mathrm{gt}}\|^{2}  + \\
& Tr(\Sigma_{\mathrm{gen}}+\Sigma_{\mathrm{gt}}-2(\Sigma_{\mathrm{gen}}\Sigma_{\mathrm{gt}})^{1/2}),
\end{aligned}
\end{equation}
where $\mu \in \mathbb{R}^{F}$ and $\sigma \in \mathbb{R}^{F \times F}$ are the mean and covariance matrix of perception features obtained from a pretrained motion classifier model with $F$ dimension of the perception features. $Tr(\cdot)$ denotes the trace of a matrix. In the experiment, we report the FID of generation to train ($FID_{tr}$) and test ($FID_{te}$) datasets, respectively.

    \item \textit{Action Accuracy} (ACC). To evaluate motion realism, we report the action recognition accuracy of the generated motions using the same pretrained action recognition model.
    \item \textit{Average Displacement Error} (ADE). ADE is the $L2$ distance between the predicted motion and ground-truth motion to measure the accuracy of the whole sequence:
    \begin{equation}
    ADE=\min_{i=(1, 2, 3, \cdots, N)}\frac1T\sum_{k=1}^{T}\|\hat{\mathbf{Y}}_k^i-\mathbf{Y}_k^i\|_2,
    \end{equation}
    where $Y$ is the Ground Truth motion sequence.
    $\hat{Y}$ is the predicted motion sequence. $N$ is the final number of predicted motion sequences.
    \item \textit{Average Pairwise Distance} (APD). APD is the average $L2$ distance between all the generation pairs, it is used to measure the diversity:
\begin{equation}
    \label{equ:13}
    APD=\frac1{N(N-1)}\sum_{i=1}^N\sum_{j=1,j\neq i}^N||(\hat{\mathbf{Y}})^i-(\hat{\mathbf{Y}})^j||,
    \end{equation} 
\end{enumerate}

\begin{table*}[!t]
\centering
\caption{Ablation studies for the influence of the multi-criteria generation framework on the performance of our methods. \textit{w.o.} means ``without'', \textit{w.} means ``with''. }
\label{tab:AblationDiffGen}
\resizebox{15.5cm}{!}{
\begin{tabular}{cc|ccccc|ccccc}
\hline
\multicolumn{2}{c|}{\multirow{2}{*}{Method}}                          & \multicolumn{5}{c|}{BABEL}                                                                     & \multicolumn{5}{c}{HAct12}     \\ 
\cline{3-12} 
\multicolumn{2}{c|}{}                                                 & $FID_{tr}\downarrow$ & $FID_{te}\downarrow$ & $ACC\uparrow$ & $ADE\downarrow$ & $APD\uparrow$ & \multicolumn{1}{c}{$FID_{tr}\downarrow$} & \multicolumn{1}{c}{$FID_{te}\downarrow$} & \multicolumn{1}{c}{$ACC\uparrow$} & \multicolumn{1}{c}{$ADE\downarrow$} & \multicolumn{1}{c}{$APD\uparrow$} \\ \hline
\multirow{2}{*}{MCG-IMM (VAE)}  & \textit{w.o.} EMG & 16.45                & 18.45                & 73.14         & 1.42            & 3.46          & \multicolumn{1}{c}{134.27}               & \multicolumn{1}{c}{147.62}               & \multicolumn{1}{c}{58.14}         & \multicolumn{1}{c}{1.36}            & \multicolumn{1}{c}{0.86}          \\
                                & \textit{w.} EMG   & 14.13                & 15.71                & 74.21         & 1.12            & 6.14          & \multicolumn{1}{c}{115.36}               & \multicolumn{1}{c}{135.26}               & \multicolumn{1}{c}{67.27}         & \multicolumn{1}{c}{1.21}            & \multicolumn{1}{c}{1.81}          \\ \hline
\multirow{2}{*}{MCG-IMM (GAN)}  & \textit{w.o.} EMG & 16.13                & 17.16                & 73.56         & 1.36            & 4.02          & \multicolumn{1}{c}{125.86}               & \multicolumn{1}{c}{137.94}               & \multicolumn{1}{c}{66.13}         & \multicolumn{1}{c}{1.35}            & \multicolumn{1}{c}{1.74}          \\
                                & \textit{w.} EMG   & 14.46                & 14.49                & 75.16         & 1.10            & 6.96          & \multicolumn{1}{c}{116.36}               & \multicolumn{1}{c}{130.49}               & \multicolumn{1}{c}{68.91}          & \multicolumn{1}{c}{1.20}            & \multicolumn{1}{c}{1.92}          \\ \hline
\multirow{2}{*}{MCG-IMM (DDPM)} & \textit{w.o.} EMG & 15.51                & 14.61                & 75.04         & 1.07            & 6.84          & \multicolumn{1}{c}{124.65}               & \multicolumn{1}{c}{131.74}               & \multicolumn{1}{c}{67.46}          & \multicolumn{1}{c}{1.36}            & \multicolumn{1}{c}{1.84}          \\
                                & \textit{w.} EMG   & 14.32                & 13.65                & 77.32         & 1.04            & 8.45          & \multicolumn{1}{c}{116.16}               & \multicolumn{1}{c}{121.34}               & \multicolumn{1}{c}{70.26}          & \multicolumn{1}{c}{1.17}            & \multicolumn{1}{c}{2.45}          \\ \hline
\end{tabular}
 }
\end{table*}

\begin{table*}[!htpb]
\centering
\caption{Ablation studies for the influence of the intra-class difference. \textit{w.o.} means ``without'', \textit{w.} means ``with''. }
\label{tab:AblationIntra}
\resizebox{15.5cm}{!}{
\begin{tabular}{cc|ccccc|ccccc}
\hline
\multicolumn{2}{c|}{\multirow{2}{*}{Method}}                          & \multicolumn{5}{c|}{BABEL}                 & \multicolumn{5}{c}{HAct12}  \\ 
\cline{3-12} 
   &  & $FID_{tr}\downarrow$ & $FID_{te}\downarrow$ & $ACC\uparrow$ & $ADE\downarrow$ & $APD\uparrow$ & $FID_{tr}\downarrow$ & $FID_{te}\downarrow$ & $ACC\uparrow$ & $ADE\downarrow$ & $APD\uparrow$ \\ \hline
\multirow{2}{*}{MCG-IMM (VAE)}  & \textit{w.o.} intra-class & 15.94                & 17.04                & 73.46         & 1.34            & 5.97          & 129.27               & 142.64               & 61.14         & 1.35            & 1.58          \\
                                & \textit{w.} intra-class   & 14.13                & 15.71                & 74.21         & 1.12            & 6.14          & 115.36               & 135.26               & 67.27         & 1.21            & 1.81          \\ \hline
\multirow{2}{*}{MCG-IMM (GAN)}  & \textit{w.o.} intra-class & 16.24                & 16.42                & 73.95         & 1.29            & 6.46          & 124.36               & 133.17               & 66.13         & 1.29            & 1.76          \\
                                & \textit{w.} intra-class   & 14.46                & 14.49                & 75.16         & 1.10            & 6.96          & 116.16               & 130.49               & 68.91         & 1.20            & 1.92          \\ \hline
\multirow{2}{*}{MCG-IMM (DDPM)} & \textit{w.o.} intra-class & 15.03                & 14.19                & 76.36         & 1.06            & 7.49          & 121.45               & 129.58               & 67.42         & 1.64            & 2.14          \\
                                & \textit{w.} intra-class   & 14.32                & 13.65                & 77.32         & 1.04            & 8.45          & 116.16               & 121.34               & 70.26         & 1.17            & 2.45          \\ \hline
\end{tabular}
}
\end{table*}

\subsection{Baseline Methods}
In this work, we compare the proposed method with the state-of-the-art human motion in-betweening generation methods as baseline methods, including RMI \cite{harvey2020robust}, MITT \cite{qin2022motion}, Motion DNA \cite{zhou2020generative}, ACTOR \cite{petrovich2021action}, WAT \cite{mao2022weakly}, MultiAct \cite{lee2023multiact}, MoFusion \cite{dabral2023mofusion}. 
%

\section{Results and Analysis}
\label{Results}

\subsection{Quantitative Results }
Table \ref{tab:sota} summarizes the comparative performance of the proposed MCG-IMM method against various baselines for in-betweening motion generation task.
The results demonstrate that the MCG-IMM method achieves state-of-the-art performance across all evaluation metrics.
Specifically, the diversity metric (APD) improvements are more obvious than the accuracy metrics.
This superiority highlights the effectiveness of the proposed framework. By formulating the in-betweening task as a multi-criteria problem, MCG-IMM can generate motion sequences that satisfy the proposed different criteria.
Interestingly, the MCG-IMM method exhibits relatively weaker performance on the NTU RGB-D dataset compared to other datasets. This is likely due to the higher level of noise and artifacts present in NTU RGB-D, which poses additional challenges for MCG-IMM.
On the other datasets, the MCG-IMM method demonstrates substantial improvements in both the accuracy and diversity metrics. This is a direct result of the multi-criteria optimization approach, which can explore and provide a diverse set of solutions between the competing criteria.
Overall, the comprehensive evaluation on multiple human motion datasets shows the superiority of the proposed MCG-IMM method for in-betweening human motion generation task. The multi-criteria optimization function proves to be a highly effective technique for this task.

In addition, the results in Table \ref{tab:sota} also report the performance of MCG-IMM using different generative models. Specifically, the table compares MCG-IMM when using VAE, GAN, and DDPM as the base generative models.
The results reveal that the MCG-IMM approach consistently outperforms other in-betweening methods regardless of the specific generative model. This indicates that the multi-criteria optimization method introduced in MCG-IMM effectively enhances performance of the generated human motions.

\begin{figure}
    \centering
    \includegraphics[width=1\linewidth]{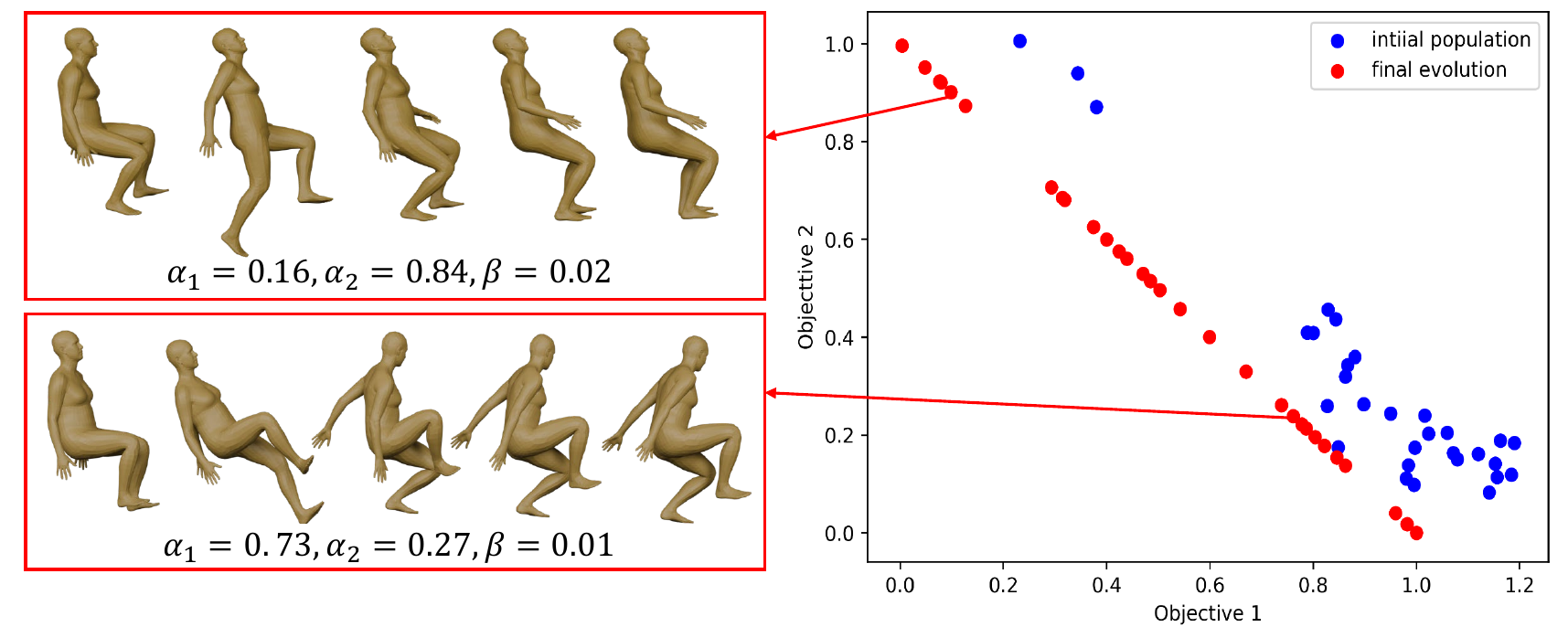}
    \caption{Pareto front of the optimization process under the ``Sit'' motion sequences.
    The blue points denote the initial samples. The red points denote the final samples. $\alpha$ and $\beta$ denote the value in multi-criteria optimization problem.
    }
    \label{fig:pareto}
\end{figure}

\subsection{Ablation Studies}
The paper conducts ablation studies to analyze the contributions of the different components within the proposed MCG-IMM method. 
As shown in Table \ref{tab:AblationDiffLength}, experiments compare the influence of different methods under variable and fixed length of transition sequences, including RMI \cite{harvey2020robust}, MITT \cite{qin2022motion}, MultiAct \cite{lee2023multiact} and CMIB \cite{kim2022conditional}.
The comparison results clearly show that allowing variable length transition sequence generation leads to enhanced performance compared to the fixed length. This indicates the importance of flexibility in the length of in-betweening motions for achieving high-quality results.

Table \ref{tab:AblationDiffGen} aims to assess the impact of the introduced multi-criteria guidance.
Specifically, we evaluated the performance of the MCG-IMM combined with various generative models, including VAE, GAN, and DDPM. The results show that the inclusion of the MCG-IMM consistently improves performance, particularly in terms of the diversity of the generated motion sequences within a bath process.
These ablation findings further report the effectiveness of our design choices.

To enhance diversity within the same motion category, the proposed multi-criteria formulation explicitly incorporates intra-class variations in the diversity component. This is achieved by utilizing the predicted class probabilities from a motion classifier, which are integrated into the objective function via the term $P_c(Y)$ in Equation \ref{EqMulObj1}. The design is applicable across different generative models, guaranteeing that both inter-class and intra-class differences contribute to the diversity of the generated motion sequences.
Table \ref{tab:AblationIntra} reports the influence of intra-class difference.
From the comparison results, we can observe that introducing this intra-class difference term into the optimization process can effectively enhance the performance of the generated motion sequences. This finding highlights the importance of explicitly modeling the inherent variations within the same motion class.

\begin{figure}
    \centering
    \includegraphics[width=1\linewidth]{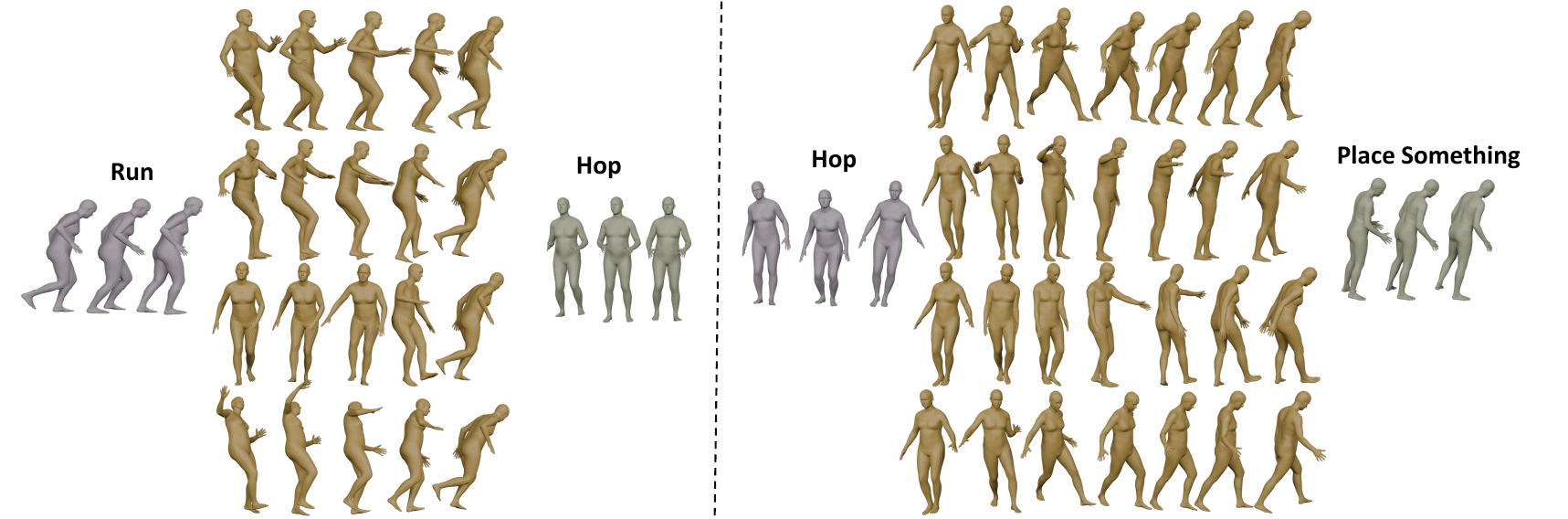}
    \caption{The generated diverse in-betweening human motion sequences given the provided conditions.}
    \label{fig:inbetweeing}
\end{figure}

\subsection{Qualitative Results}
In the qualitative analysis, we first visualize the Pareto front obtained from the multi-criteria optimization process.
Fig. \ref{fig:pareto} plots the distribution of the 20 human motion sequences in the samples. The blue points denote the initial samples, which are the generated motion sequences without the multi-criteria guidance. The red points represent the generated in-betweening human motion sequences after the multi-criteria guidance.
As shown in this figure, the initial samples are unable to approximate the Pareto front effectively, and the distribution of the solutions is rather irregular. In contrast, the MCG-IMM approach is much more effective in locating diverse optimal solutions along the Pareto front.
In addition, we show two points on the Pareto front that are far apart, which correspond to two human motion sequences. This demonstrates its superior ability to balance the trade-offs between the competing objectives during the in-betweening human generation task.

In addition, we also visualize the in-betweening human motion sequences given the user-provided motion sequences.
As illustrated in Fig. \ref{fig:inbetweeing}, the left side of the figure shows the transition from the ``Run" to the ``Hop" action, while the right side displays the transition from the ``Hop" to the ``Place" action. Since the two actions on the left have a small difference, the generated transition motion sequence is relatively short. The comparison results show that our method can generate variable lengths in-betweening human motion sequences.
These qualitative results demonstrate that our method can generate more diverse human motion sequences with the introduced multi-criteria guidance.

\section{Conclusion}
\label{conclusion}

This paper presents MCG-IMM, a plug-and-play multi-criteria guidance framework for generating diverse and smooth in-betweening human motions. Unlike prior approaches that rely on architectural modifications or additional training procedures, MCG-IMM operates entirely at the sampling stage of pretrained generative models—introducing no extra parameters and requiring no retraining. By formulating the in-betweening task as a multi-criteria optimization problem, MCG-IMM simultaneously accounts for motion diversity and smoothness, and efficiently solves this problem via evolutionary optimization techniques.
MCG-IMM is model-agnostic and compatible with various generative backbones, including DDPM, VAE, and GAN. Empirical evaluations across four widely used human motion datasets demonstrate that our method consistently improves both intra-batch diversity and temporal consistency, outperforming existing state-of-the-art approaches. Extensive studies further validate the individual contributions of each component, including intra-class variation modeling and variable-length generation.





\bibliographystyle{unsrt}
\bibliography{cite}

\vfill

\end{document}